\documentclass[review, 3p]{elsarticle}

\usepackage{mypreamble}
\usepackage{graphicx}        
\usepackage{amsmath}
\usepackage{amsthm}
\usepackage{amssymb}
\usepackage{mathtools}
\usepackage{epstopdf}
\usepackage{float}
\usepackage{epsfig}
\usepackage{wrapfig}
\usepackage{amsfonts}
\usepackage{dcolumn}% Align table columns on decimal point
\usepackage{bm}% bold math
\usepackage{placeins}
\usepackage{verbatim}
     
\graphicspath{ {./figures/} } 

\theoremstyle{plain}

\newtheorem{thm}{Theorem}

\theoremstyle{definition}
\newtheorem{example}{Example}
\newtheorem*{pf}{Proof}
\newtheorem*{remark}{Remark}

\newcommand{\real}{\mathbb{R}}

\makeatletter
\def\ps@pprintTitle{%
 \let\@oddhead\@empty
 \let\@evenhead\@empty
 \def\@oddfoot{}%
 \let\@evenfoot\@oddfoot}
\makeatother

\begin{document}

\begin{frontmatter}

\title{A Stochastic Binary Opinion Model: Opinion Dominance vs. Balance}
\author[1]{Serap Tay Stamoulas\corref{cor1}%
\fnref{fn1}}
\ead{serap.taystamoulas@dicle.edu.tr}

\author[2]{Muruhan Rathinam\fnref{fn2}}
\ead{muruhan@umbc.edu}  

\cortext[cor1]{Corresponding author}

%\fntext[fn1]{This is the first author footnote.}
%\fntext[fn2]{Another author footnote, this is a very long
%footnote and it should be a really long footnote. But this
%footnote is not yet sufficiently long enough to make two
%lines of footnote text.}
%\fntext[fn3]{Yet another author footnote.}

\affiliation[1]{organization={Department of Mathematics, Faculty of Science, Dicle University},
%addressline={},
city={Diyarbak\i r},
postcode={21280},
country={Turkey}}

\affiliation[2]{organization={Department of Mathematics and Statistics, University of Maryland, Baltimore County},
%addressline={1000 Hilltop Circle},
city={Baltimore, MD},
postcode={21250},
country={USA}}

\begin{abstract}
We propose and study a stochastic binary opinion model where agents in a group are
considered to hold an opinion of 0 or 1 at each moment. An agent in the group
updates his/her opinion based on the group's opinion configuration and his/her
\emph{personality}. Considering the number of agents with opinion 1 as a
continuous time Markov process, we analyze the long-term probabilities for
large population size in relation to the personalities of the group. In
particular, we focus on the question of ``balance'' where both opinions are
present in nearly equal numbers as opposed to ``dominance'' where one opinion is present in a greater number. 
\end{abstract}


\end{frontmatter}

    %! Author = sbbfti
%! Date = 10/06/2020

% Document

\section{Introduction}
The study of opinion dynamics focuses on modeling the decision-making process in
multi-agent systems. 
%This line of research  dates back to 1950s \cite{french1956formal}. 
 It is natural to assume that an agent's decision-making process is influenced by the information received from the society. The influence of society could be considered in the form of agent's interactions with his/her  ``neighbors'' where neighbors may include every other agent in the society \cite{mobilia2007role,nyczka2013anticonformity}, may be determined based on the agent's opinion  \cite{stamoulas2018convergence,blondel2010continuous,ceragioli2012continuous} or based on a given communication graph independent of the opinions \cite{degroot1974reaching,jia2015opinion,friedkin1990social,bartolozzi2005stochastic,holyst2000phase,yildiz2013binary,sznajd2000opinion}.  
 %hendrickx2016symmetric removed
The set of possible opinions of a given 
agent at a given time may be regarded as a finite discrete set without any
additional structure, the simplest example being a binary set \cite{mobilia2007role,holyst2000phase,yildiz2013binary,sznajd2000opinion,san2005binary,galam2013modeling,zanette2006opinion,lyst2002social}
or alternatively as a continuum of values in $\real^d$ \cite{stamoulas2018convergence,blondel2010continuous,hegselmann2002opinion,weisbuch2002meet}. 
%The structure of the $\real^d$ imposes a concept of nearby opinions 
%and leads to the notion of confidence bounds where an agent may only interact 
%with other agents within his/her confidence bound. See for instance \cite{hegselmann2002opinion,blondel2010continuous,hendrickx2016symmetric,stamoulas2018convergence,ceragioli2012continuous}. In this paper, we shall focus on the binary opinion case and as such the notion of bounded confidence does not apply. 
%See \cite{bartolozzi2005stochastic,holyst2000phase,yildiz2013binary,sznajd2000opinion,holley1975ergodic,san2005binary,galam2013modeling,zanette2006opinion,mobilia2007role,lyst2002social} for binary opinion models.

The celebrated voter model \cite{liggett1999stochastic} provides a basic updating rule for an agent holding an opinion from a binary set. Based on this updating rule, at each time step an agent will update his/her opinion by conforming to the majority opinion of his/her neighbors. Here, neighbors of an agent are defined through a network structure.  It is important to note that this opinion updating rule assumes that agents conform to the neighbor(s) and hence it does not take agents' personalities into account. In this respect, \cite{mobilia2007role,yildiz2013binary,galam2013modeling,schneider2004influence,sznajd2011phase,de2005spontaneous, javarone2014social,khalil2018zealots} and references therein include agents with \emph{personalities}  who can choose not to conform to the majority opinion of their neighbors and study the effects of such agents on the group's limiting decision configuration.
The presence of  \emph{stubborn} agents who do not  change their opinion is studied in \cite{yildiz2013binary,galam2013modeling,mobilia2007role,khalil2018zealots}. \emph{Independent} agents who change their opinions independent of the interactions are considered in \cite{sznajd2011phase}.
\cite{galam2013modeling,schneider2004influence,de2005spontaneous} study  \emph{contrarian} agents who adopt the opposite of the leading opinion with a certain probability.  
%The presence of \emph{leaders} is considered in \cite{ellero2013stochastic,lyst2002social,holyst2000phase}. 
In our model, we capture different personality traits through a monotonic \emph{conformity} function and a \emph{spontaneity} coefficient.

It is natural to focus on two possible extreme outcomes:
balance of opinions where both opinions are present in nearly equal numbers and consensus where all agents agree. In reality, the situation is not  ``black
or white''; for instance, one may find that in the long-term, one opinion is likely
to be held by say $70\%$ of the agents resulting in dominance of one opinion. 
The idea of balance of opinions and dominance of opinions are explored in \cite{galam2013modeling,mobilia2007role,sznajd2000opinion} in discrete time setting in relation to the personalities of agents. 

In this study, we propose a continuous time stochastic binary opinion model (say $0$ or $1$) for a group in which agents
are considered to have personalities defined by a monotonic conformity function and a
spontaneity coefficient. An agent's personality determines the effect of
social influence on the agent (e.g., conformists, rebellious). Moreover, in our model, contrary to the pair-wise
interaction of agents that are defined to be neighbors, agents are considered
to be informed on the distribution of opinions in the entire group at each time
$t$. We investigate various personality traits and their effect on the group's limiting behavior for a large number of agents. 
In particular, the personality traits that lead to dominance of one opinion is our main interest.
We note that, our model can be thought of as the result of a situation 
where agents do not change their mind after one (pairwise or group) interaction, but rather after 
several interactions. In this case, assuming all agents can interact with all 
others, in a large population an agent interacts with
sufficiently many others before changing their mind and the sample from the 
sufficiently many can be taken as a good approximation of sampling the entire
population. 

We also assume (as is done in other models) that there is no ``natural bias'' towards one of the opinions. 
As an example, if the opinion is about whether the earth is flat or not, one
would expect that in an informed society, agents are more likely to believe
the truth. We are focused on the long-term probability distribution for 
the number of agents with opinion $1$ when the number $N$ of total agents is
very large. In particular, we study the effects of agents' personalities on balance 
and dominance of the opinions. We found that the shape of the conformity function plays an important role.

The paper is organized as follows. In Section \ref{sec:model} we introduce our
binary model and focus on a \emph{homogeneous} group. In section \ref{sec:C1}
we study the effects of personality of the (homogeneous) group on the group's
limiting behavior. We extend our model to \emph{heterogeneous} groups in \ref{sec:heterogeneous} and
examine limiting decision behavior when the group is formed  by  two extreme
personality classes in Section \ref{sec:twoPersonality}. The concluding remarks follow in Section \ref{sec:conclusions}.

\section{\label{sec:model}The model and the homogeneous case}
We consider a group of $N$ agents where each agent holds an opinion from the set $\{0,1\}$. 
An agent flips his/her opinion based on the group's current configuration and his/her  \emph{personality}. Here,  \emph{personality} of an agent $i$, $i = 1,\dots,N$,  is given by the pair $(\phi_i, \beta_i)$  where 
 $\phi_i: [0,1]\to[0,\infty)$ is a monotonic function that accounts for {\em conformity}, $\beta_i$ is a 
nonnegative quantity that accounts for {\em spontaneity}. We call a group
\emph{homogeneous} if all agents in the group share the same personality,
$\phi = \phi_i, \beta = \beta_i \; \forall i= 1,\dots,N$. We first look at the
case when the group is homogeneous. 

We define
$X^N(t)$ to be the number of agents holding opinion 1
 at time $t \in [0,\infty)$ and assume that a given agent 
flips his/her opinion during a time interval $(t,t+h]$ with a 
probability
\begin{equation}
\label{eq:opinion-change-rate}
(\phi(n/N) + \beta) h + o(h)\; \text{as}\; h \to 0+,
\end{equation}
where $n$ is the number of
agents with opposite opinion and $N$ is the total number of agents at time
$t$.
Thus we note that $\phi(x)$ determines the rate of conformity where $x$ is the
fraction of the population holding the opposite view. We note that the
pairwise interaction model is a special case of our model where $\phi$ is
linear, that is
$\phi(x)=\alpha x$ for some $\alpha>0$. 
 
This results in a continuous time Markov process model for $X^N(t)$ with the state space  $\{0,1,2,\ldots,N\}$. 
 Moreover,  we may regard $X^N(t)$ as a birth-death process since when an agent flips his/her opinion, the
process  increases/decreases by one. 
Using the rate of opinion change for an agent defined by \eqref{eq:opinion-change-rate}, 
 the birth rate $\lambda_i^N$ and the death rate  $\mu_i^N$ at the state $X^N(t) = i$ can be written as follows:
\begin{eqnarray}
\label{eq:jump-up}
\lambda_i^N
&=& \left(\phi \left(\frac{i}{N} \right) + \beta \right)(N-i)
=          (N-i)  \phi \left(\frac{i}{N} \right)+\beta (N-i), \nonumber \\
\mu_i^N 
\label{eq:jump-down}
&=& \left( \phi \left(\frac{N-i}{N} \right) + \beta \right) i
=      i \phi \left( \frac{N-i}{N} \right) + \beta i.
\end{eqnarray}
%\begin{eqnarray}
%\label{eq:jump-up}
%\lambda_i^N
%&=& q_{i,i+1}=(\frac{i}{N-i}\alpha + \beta )(N-i)
%=          \alpha X(t) +\beta (N-X(t)), \\
%\mu_i^N 
%\label{eq:jump-down}
%&=& q_{i,i-1}=(\frac{N-i}{i}\alpha + \beta)i
%=     \beta X(t)+\alpha (N-X(t)),
%\end{eqnarray}
Since the state space $\{0,1,2,\ldots,N\}$ is finite, $\lambda_N^N=0$ and $\mu_0^N=0$. 
Using these transition rates one can construct a transition rate matrix
$Q^N=[q_{ij}^N]$,  where $q_{i( i+1)}^N = \lambda_i^N$, $q_{i( i-1)}^N = \mu_i^N$, $q^N_{ii}= - \sum_{j\ne i} q^N_{ij}$
and $q_{ij}^N = 0 \; \forall j \notin \{i-1,i, i+1\}$. 

We may rewrite the birth rate $ \lambda_i^N $  and  the death rate  $\mu_i^N$ as follows:
\begin{equation}
\begin{aligned}
\label{eq:bd-rates}
        \lambda_i^N &=  N \bar{\lambda} \left(\frac{i}{N} \right), \quad  \\
             \mu_i^N &= N\bar{\mu}\left(\frac{i}{N} \right) \quad i =0, 1,\dots, N,
\end{aligned}
\end{equation}
%\begin{eqnarray}
%\label{eq:bd-rates}
%         \lambda_i^N &=&  N \bar{\lambda}(\frac{i}{N}), \quad \lambda_N^N = 0, \quad i = 0,1,\dots, N-1,	\nonumber \\
%             \mu_i^N &=&  N\bar{\mu}(\frac{i}{N}), \quad \mu_0^N = 0, \quad i = 1,\dots, N
%\end{eqnarray}
where $\bar{\lambda},\bar{\mu}:[0,1] \rightarrow  [0,\infty)$ are
 %$C^1([0,1])$ functions and
  given by
\begin{equation}
\begin{aligned}
\label{lamda_mu_bar}
	\bar{\lambda}(x) =&  (1-x)  \phi(x)+ \beta (1-x), \\
	\bar{\mu}(x) =&    x \phi(1-x)+  \beta x.
	\end{aligned}
\end{equation} 
We define the probability   $p^N(t)=(p^N_0(t), p^N_1(t), \ldots, p^N_N(t) )$,
where  $p^N_i(t)=\mathbb{P}[X^N(t)=i]$ for each $i=0,1,\ldots,N$. The
probability mass function satisfies the Kolmogorov's forward equation 
%$
%\dot{P}^N_{ij}(t) = \sum_k P^N_{ik}(t) q^N_{kj}
%$, 
\begin{equation}
\label{eq:p-derivative}
\dot{p}^N_j(t) =     \sum_k p^N_k(t)q^N_{kj}.
\end{equation}
It should be noted that when \emph{spontaneity} coefficient $\beta =0$, 
the states $i =0$ and $i= N$ are absorbing states since 
$\lambda_0^N = \mu_N^N
= 0$. When $\beta>0$, the birth-death process $X^N(t)$ is an irreducible
Markov process with the finite state space $\{0,1,2,\ldots,N\}$ and thus,
$X^N(t)$ is ergodic and  attains a  unique stationary probability distribution
as $t \to \infty.$ The probability vector  $p^N(t) \to \pi^N =
(\pi^N_0,\pi^N_1,\ldots,\pi^N_N)$ as  $t\rightarrow \infty$  and $\pi^N$ does
not depend on the initial state $X^N(0)$.  Using the detailed balance condition at stationarity, one can obtain
\begin{equation*}
                   \pi^N_n= \frac{\lambda^N_{n-1}\lambda^N_{n-2}\ldots\lambda^N_0}
                                {\mu^N_n \mu^N_{n-1}\ldots \mu^N_1} \pi_0^N, \quad n=1,2,\ldots,N.
\end{equation*}
Hence,
\begin{eqnarray}
\label{eq:pi_n}
        \pi_n^N &=& \frac{R^N_n}{\sum_{k=0}^N R^N_k}, \quad n=0,1,\ldots,N. \\
\label{eq:pi_0}
         \pi_0^N &=& \frac{1}{\sum_{k=0}^N R^N_k}, 
\end{eqnarray}
where  
\begin{equation}
\label{eq:R_n^N}
R_n^N = r^N_n r^N_{n-1}\ldots r^N_1, \quad n = 1,2,\ldots,N
\end{equation}
 for  $r_n^N  = \frac{\lambda^N_{n-1}}{\mu^N_n}$  and $R^N_0 = 1$.
 
%Our interest is the nature of the stationary probabilities with respect to the personality of the group. In particular, stationary probabilities are expected to illustrate either one of two scenarios: balance of opinions or  dominance of one opinion.  In the next sections, we will provide an analysis of \eqref{eq:pi_n} and \eqref{eq:pi_0} for large $N$ in this perspective with respect to the conformity function, $\phi(x)$ and spontaneity coefficient $\beta$.

\section{\label{sec:C1}The effects of the conformity function : The homogeneous case }

%In this section, we study various conformity functions $\phi(x)$ and analyze 
%their effects on the behavior of $X^N(t)$ for large $N$ and large $t$. 
%When the process $X^N(t) = i$, the birth and death rates are as in  \eqref{eq:bd-rates}, where
% $\bar{\lambda},\bar{\mu}:[0,1] \rightarrow  [0,\infty)$ are class of  $C^1([0,1])$  and given by \eqref{lamda_mu_bar}.
%\begin{eqnarray}
%\label{eq:rate-function-general}
%	\bar{\lambda}(x) &=&  (1-x) [ \phi(x)+ \beta ] ,	\nonumber \\
%	\bar{\mu}(x) &=&    x [ \phi(1-x)+  \beta ].
%\end{eqnarray} 
In order to study $X^N(t)$ for large $N$ and $t$, we shall consider the normalized process $X_N(t) = \frac{X^N(t)}{N}$. As $N \to \infty$, in the fluid limit, one expects $X_N$ to converge to $x$ 
where $x$ satisfies the ODE 
\begin{equation}
\label{eq:binary-ode-gen}
\dot{x}(t)= F(x(t)) =  \bar{\lambda}\big(x(t)\big) - \bar{\mu}\big(x(t)\big),
\end{equation}
where
\begin{equation}
\label{eq:vectorfield-homo}
F (x) =   \phi\left(x\right) \left( 1-x \right) - x \phi\left(1-x \right) +
         \beta (1 - 2 x).
\end{equation}
Intuitively, when $N$ and $t$ are both large, one expects the peaks of the probability 
distribution of $X_N(t)$ to occur near the stable equilibria of this ODE. 
This observation will motivate the rest of the analysis in this paper. 

While we do not make new claims about rigorous limits as $t \to \infty$ and $N \to \infty$ jointly, some rigorous limits exist in literature that we 
mention here. A major result is that if $F$ is $C^1$ (continuously differentiable), then  
%and \eqref{eq:binary-ode-gen} has a solution $X(t)$ on 
given any finite time interval $[0,T]$, $X_N \to x$ uniformly 
on $[0,T]$ with probability one as $N \to \infty$. Moreover a diffusion approximation
for $X_N$ is also available \cite{ethier2009markov}. Since this result
only considers the limit as $N \to \infty$ over finite intervals of time, 
one needs to be cautious in interpreting the large $N$ and large $t$
approximation. In particular, if one fixes any large final time $t$, and considers 
increasing $N$, then one expects distributions at time $t$ to have peaks
around the stable equilibria of the ODE.   
When $F$ has a unique globally attractive 
equilibrium $\overline{x}$, under suitable conditions, as $N \to \infty$ one can
rigorously justify a Gaussian approximation with mean $\overline{x}$ for the
stationary probability distribution (see Theorem 2.7 in  \cite{kurtz1976limit}).

%\begin{theorem}[Law of Large Numbers ~\cite{ethier2009markov}]
%\label{LLN}
%Suppose for each compact $K \subset \real$, $F(x) = \bar{\lambda}(x) - \bar{\mu}(x)$ is Lipschitz on $K$, that is, for each $x,y \in K$, there exists some constant $M_K$ such that 
%\[
%	| F(x) - F(y)| \leq M_K|x-y|.
%\]
%Suppose $X_N(0) = x_0 \in  \real_+$ for all $N$. Let $X$ be the solution of the equation 
%\[
%	X(t) = x_0 + \int_0^t F(X(s))\;ds,
%\]
%which we assume exists on $[0,T]$. Then for all $t \in [0,T]$,
%\[
%	\lim_{N}\sup_{s\leq t} | X_N(s)-X(s)| = 0 \quad   \text{a.s.}.
%\]
%\end{theorem}
%Considering $\bar{\lambda}, \bar{\mu} :[0,1] \to [0,\infty)$  in \eqref{eq:rate-function-general}, we have
%\begin{eqnarray}
%\label{eq:binary-ode-gen}
%	 \dot{X}(t) &=& \bar{\lambda}(X(t)) - \bar{\mu}(X(t))    \nonumber \\
%	 		&=&  \ \phi(X(t)) - \Big[  \phi(X(t)) + \phi(1-X(t))  + 2\beta \Big] X(t)  + \beta \nonumber \\
%			&=& F(X(t)).
%\end{eqnarray}
%It is easy to see that when $\phi(x)$ is taken to be a $C^1$ function,  the vector field $F(X)$ in \eqref{eq:binary-ode-gen} is Lipschitz on any compact subset of $\real$. For any initial value $X_0$, the initial value problem has a unique solution  on some finite interval $ [0,T]$. On the other hand, by Theorem \ref{LLN} we can conclude that  $ X_N(t) \to X(t)$, almost surely, uniformly in $t \in [0,T]$. %Define $p^N_i(t) = \mathbb{P}[X^N(t) = i]$ and $p^N(t) = (p^N_0(t), \dots, p^N_N(t))$ and $\pi^N = \lim_{t \to \infty} p^N(t)$. 
%
Henceforth, we shall study the system \eqref{eq:binary-ode-gen} for its
stable equilibria. We note that $F(0)>0$, $F(1)<0$ and $\overline{x}=\frac{1}{2}$ is always an
equilibrium for the dynamics \eqref{eq:binary-ode-gen}. If
$\overline{x}=\frac{1}{2}$ is the unique equilibrium, it will be globally
attractive on $[0,1]$. Then, for very large $N$, the stationary probability
distribution will be narrow and approximately Gaussian with mean $\overline{x}$ and hence, the
model leads to balance of opinions. 

We shall see that the shape of the conformity function $\phi$ plays an
important role in deciding if dominance of an opinion is likely.  
Intuitively, one may expect that greater conformity leads to dominance of one
opinion while greater spontaneity leads to balance of opinions via a law of
large number effect. However, our examples suggest that the dependence on 
$\phi$ is more subtle in that the shape of the function plays a crucial role. 
We see that when $\phi$ is strictly convex, for sufficiently small $\beta$ 
dominance is observed. When $\phi$ is not strictly convex or if it is concave, we do not 
see dominance in our examples.

\begin{remark} For the sake of precision, we shall use the term {\em balance} to
mean the situation where there is only one stable equilibrium of
\eqref{eq:binary-ode-gen} which is $\overline{x}=1/2$. We shall use the 
term {\em dominance} rather loosely to stand for lack of balance.
\end{remark}
In order to investigate the effects of $\phi$, it makes sense to make some
natural assumptions. The most natural 
conditions on $\phi:[0,1] \to [0,\infty)$ are that $\phi(0)=0$ and $\phi$ is
  increasing for conformity, and $\phi(1)=0$ and decreasing for {\em rebelliousness}
  (opposite of conformity). 
\begin{thm}
\label{homogeneous_dominance}
Consider a monotone, $C^1$ conformity function $\phi(x): [0,1] \to [0, \infty)$ such that
  $\phi'(x)$ strictly increasing on $[0,1]$ and suppose that $\phi(0)=0$.
Then for sufficiently small $\beta$, the equilibrium $\overline{x}=1/2$ is unstable and hence the model leads to dominance of one opinion for large $N$ and large $t$. 
\end{thm} 
\begin{pf}
Define
$
G(x) = \phi(x) (1-x) - x \phi(1-x).
$
Then, the vector field \eqref{eq:vectorfield-homo} is 
\[
F(x) = G(x) + \beta (1-2x).
\]
We note that $G(0)=G(1/2)=G(1)=0$. 
We note that $F'(1/2) = G'(1/2) - 2 \beta$, and that 
%let us focus on computing $G'(1/2)$. 
 \[
  	G'(1/2) = \phi'(1/2) - 2 \phi(1/2).
 \]
Using the mean value theorem for $\phi(x)$ on $ [0,1/2]$, we conclude that for some $c \in (0,1/2)$, 
\[
	\phi'(c) = \frac {\phi(1/2)-\phi(0) }{1/2} = 2 \phi(1/2),
\]
since $\phi(0) = 0$. Thus,
\[
	G'(1/2) = \phi'(1/2) -\phi'(c) >0, \quad c\in(0,1/2)
\]
since $\phi '(x)$ is strictly increasing on $[0,1]$. Then $F'(1/2)=G'(1/2) - 2 \beta >0$ for sufficiently small $\beta$ and thus  $1/2$ is unstable. 
%On the other hand since $F(0)=\beta>0$, $F(1)=-\beta<0$ there must be at least two (symmetrically placed) equilibria, one in $(0,1/2)$, and the other in $(1/2,1)$.  
\end{pf}
 
 Next, we provide examples with various conformity functions $\phi$ and investigate the stable
equilibria to predict dominance or balance. We shall also check the prediction
from the stable equilibria of the ODE model 
against computational results of the stationary probability distributions. 
We note that, there are two methods to compute the stationary
distributions. One is to use the formulas \eqref{eq:pi_n} and \eqref{eq:pi_0}
and the other is to use an ODE solver to compute the solution to
\eqref{eq:p-derivative}. For very large $N$ values, our numerical experiments
suggest that using \eqref{eq:pi_n} and \eqref{eq:pi_0}  provide more accurate
results compared to the ODE solver. Hence, throughout this study, we refer to
\eqref{eq:pi_n} and \eqref{eq:pi_0} to verify our predictions from the 
stable equilibria of \eqref{eq:binary-ode-gen}.
%It is worthwhile to analyze the dynamics \eqref{eq:binary-ode-gen} for its stable equilibria. It is easy to see that $\overline{X}=\frac{1}{2}$ is always an equilibrium for the dynamics \eqref{eq:binary-ode-gen}. For a general $\phi(x)$, there can be more equilibria. 
%For example when $\phi(x) = x^2$,  $\overline{X}_1 = 0.2764$,  and  $\overline{X}_3 = 0.723$ are also equilibria. It should be noted that if  $\overline{X} =\frac{1}{2}$  is not the only equilibrium for \eqref{eq:binary-ode-gen}, then these equilibria are stable while $\overline{X}=\frac{1}{2}$  is an unstable equilibrium.

\begin{example}
Consider the simplest example of $\phi(x) =x$. This is convex, but not
strictly so. 
In this case, the rate that an agent changes his/her opinion is
$\frac{n}{N} + \beta$.  
Using \eqref{eq:binary-ode-gen}, we can conclude that as $N$ gets large $X_N(t)$ converges to $x(t)$, where 
\begin{equation}
\label{eq:ode-x}
	\dot{x}(t) = \beta ( 1 -2 x(t) ).
\end{equation}
Since this ODE has a unique equilibrium at $\overline{x} = \frac{1}{2}$ that is  globally attractive,  regardless of spontaneity coefficient $ \beta>0$,  it is expected that the group will reach balance of opinions as can be observed in Fig.~\ref{fig:exact_gen1_x}.  We note that, as can be seen in  Fig.~\ref{fig:exact_gen1_x}(b), when $\beta$ is small (in this case $\beta= 0.01$) , this results in a a wider bell shaped curve. That is,  stationary probabilities are still non zero around the equilibrium point. However, as  $N$ gets larger, the bell shaped curve becomes narrower. It is thus interesting to note that for $\beta>0$, no matter how small, one expects balance of opinions for sufficiently large $N$.

\begin{figure}[h]
\centering
\begin{subfigure}{.5\textwidth}
  \centering
  \includegraphics[width=1\linewidth]{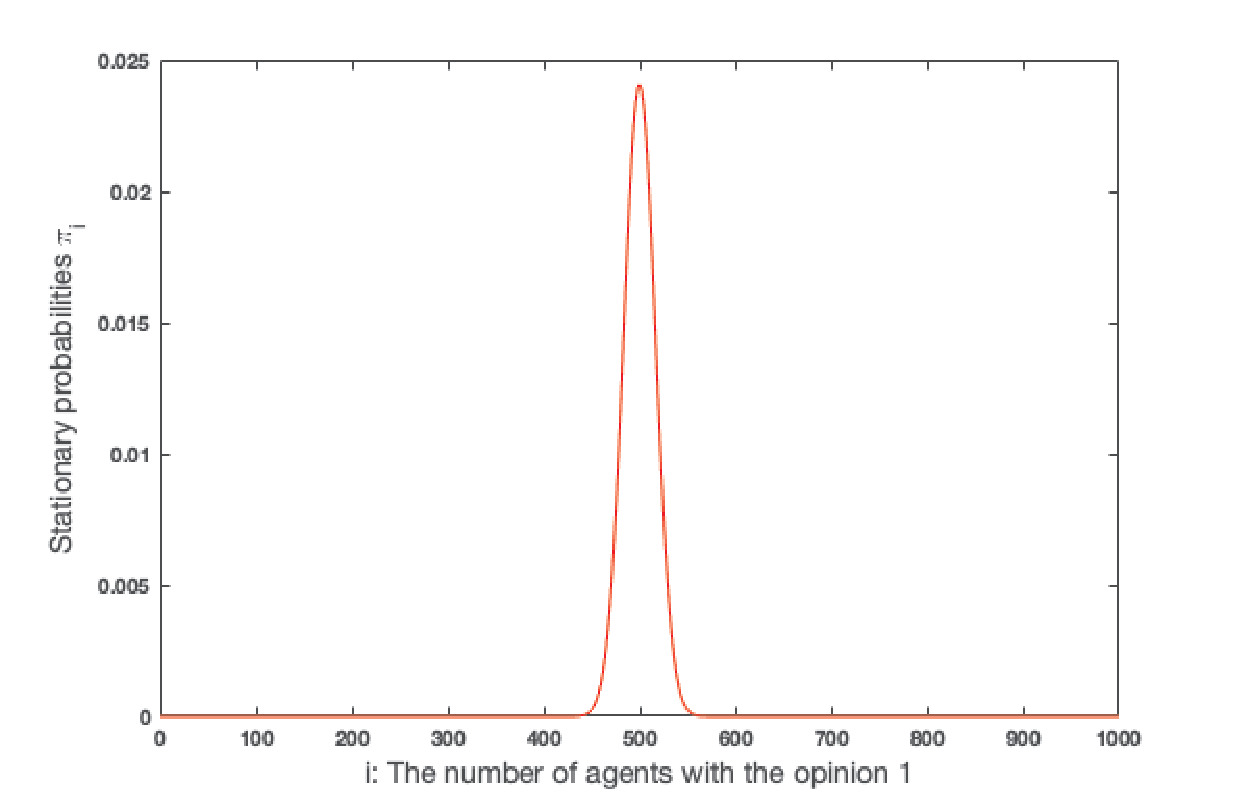}
  \caption{}
  \label{exact_gen1.eps}
\end{subfigure}%
\begin{subfigure}{.5\textwidth}
  \centering
  \includegraphics[width=1\linewidth]{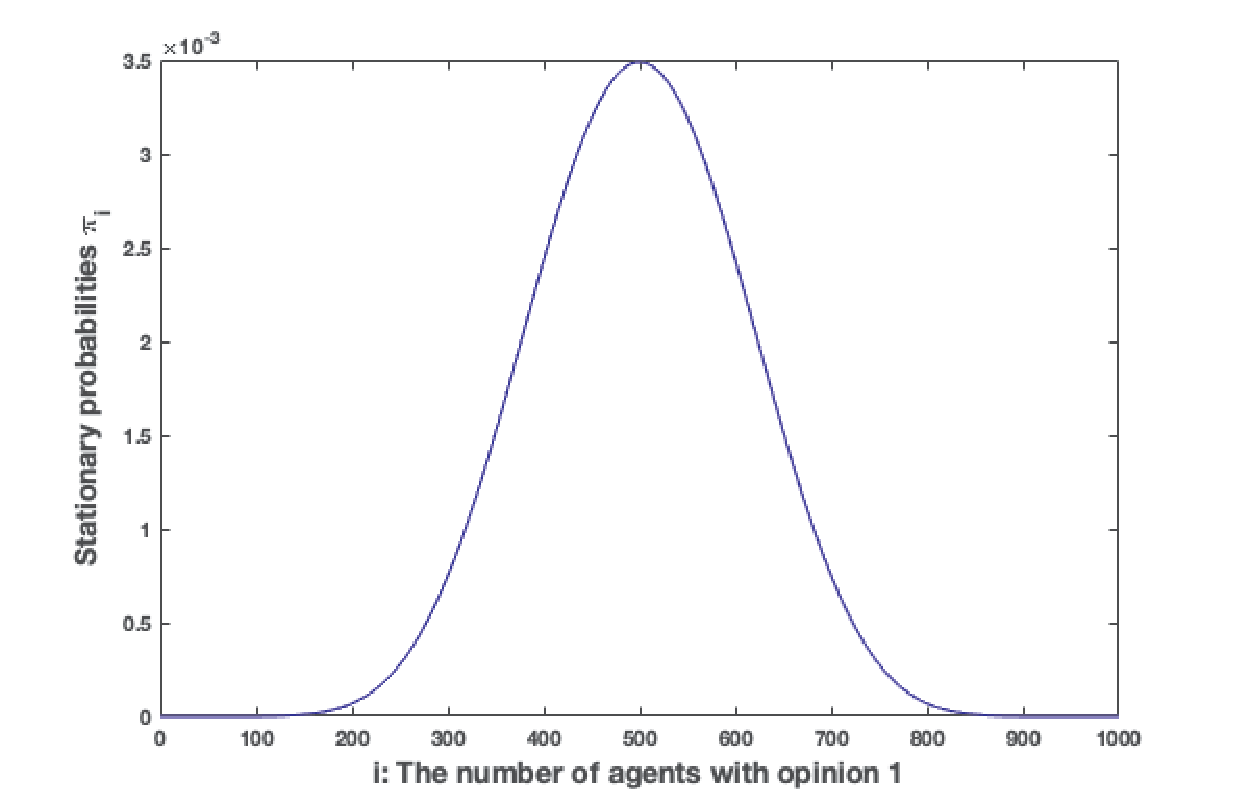}
  \caption{}
  \label{fig:exact_gen1_second.eps}
\end{subfigure}
\caption{The stationary probabilities calculated using \eqref{eq:pi_n} and \eqref{eq:pi_0}  for $\phi(x) = x$ and $N = 1000$. (a) $\beta = 5$. (b) $\beta = 0.01$. }
\label{fig:exact_gen1_x}
\end{figure}
\end{example}
\begin{example}
Consider  $\phi(x) = x^2$ which is strictly convex and $\phi(0)=0$. Thus, we can apply  Theorem \ref{homogeneous_dominance}.
%In this case, the rate at which an agent changes his/her opinion can be given as 
%$
% \frac{n^2}{N^2} + \beta	
%$ and this gives a faster rate of conformity. 
The limiting  ODE is
\begin{equation}
\label{eq:binary-ode_x2}
	\dot{x}(t) =\left(1-2x(t) \right) ( x(t)^2- x(t) + \beta).
\end{equation}
It is easy to see that \eqref{eq:binary-ode_x2}  has three possible
equilibria; $\overline{x}_1 = \frac{1}{2}$ and $\overline{x}_{2,3} =
\frac{1}{2} \pm \frac{\sqrt{1 -4 \beta}}{2 }$. Hence, based on the choice of
spontaneity coefficient $ \beta$, different scenarios are expected. When
$\beta \geq \frac{1}{4}$, $\overline{x}_1=\frac{1}{2}$ is  the only stable equilibrium and the model leads to balance of opinions. On the other hand, when $ \beta < \frac{1}{4}$, $\overline{x}_1=\frac{1}{2}$ is unstable and hence the model leads to dominance. 

In Fig.~\ref{fig:exact_gen2_x2}(a) one can observe that  the model leads to balance of opinions for $\beta = 5$ since the stationary probability distribution has its peak at the half of the population i.e., $\overline{x}_1=\frac{1}{2}$.    On the other hand, when $\beta =0.2$, the stable equilibria are $\overline{x}_{2,3} =
\frac{1\pm \frac{1}{\sqrt{5}}}{2} $, approximately $x_2 = 0.72$ and $x_3 = 0.28$. The stationary probability distribution have peaks around these equilibria as can be seen in Fig.~\ref{fig:exact_gen2_x2}(b). Thus, it is expected that  $72\%$ of the population will hold one opinion resulting in dominance of that one opinion. 
%with exact values of  $\pi_{276} = \pi_{724} = 0.0061$  
\begin{figure}[h]
\centering
\begin{subfigure}{.5\textwidth}
  \centering
  \includegraphics[width=1\linewidth]{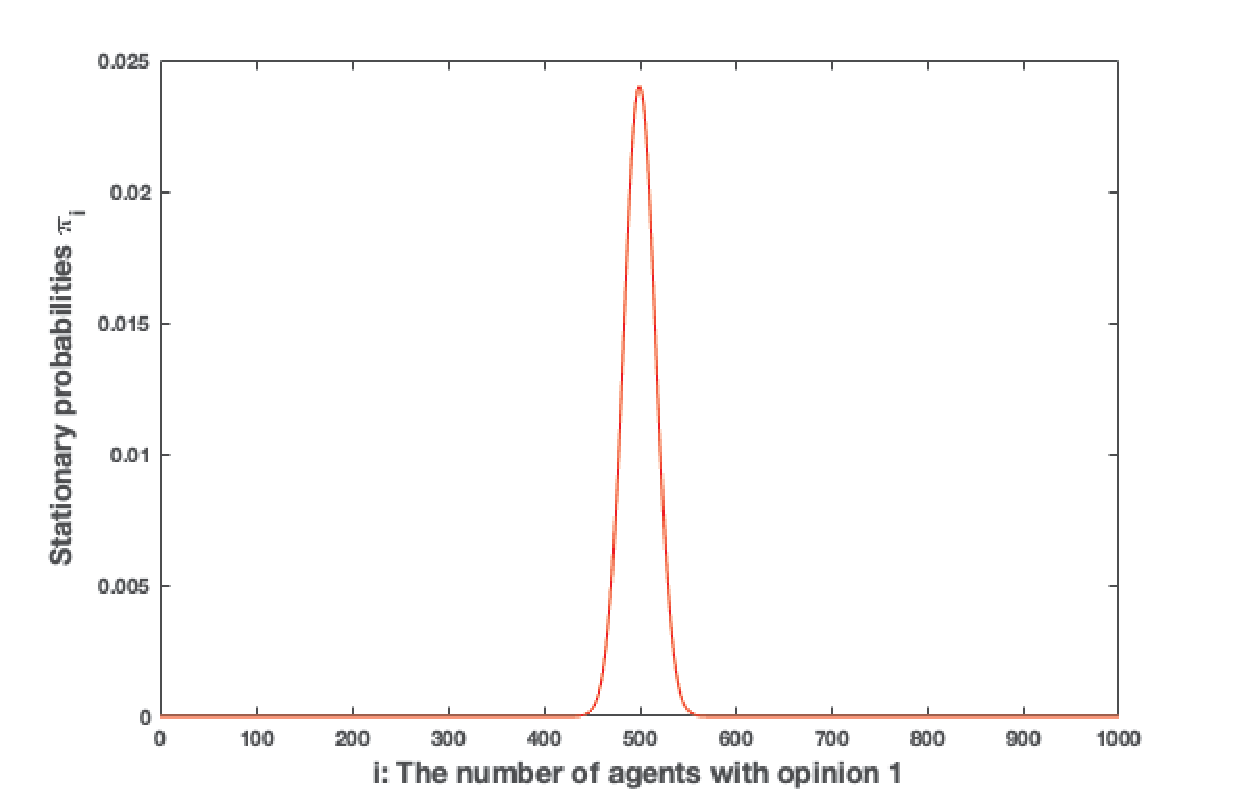}
  \caption{}
  \label{exact_gen2.eps}
\end{subfigure}%
\begin{subfigure}{.5\textwidth}
  \centering
  \includegraphics[width=1\linewidth]{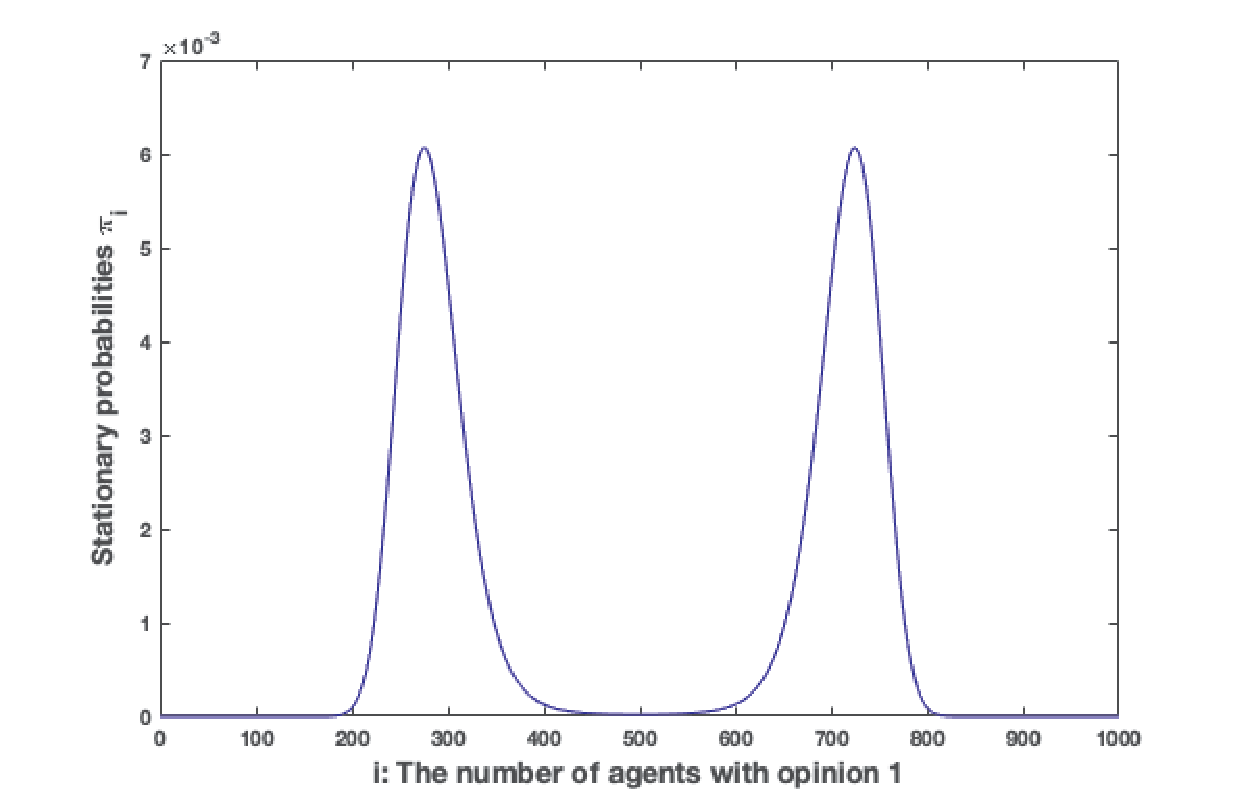}
  \caption{}
  \label{fig:exact_gen2_second.eps}
\end{subfigure}
\caption{The stationary probabilities calculated using \eqref{eq:pi_n} and  \eqref{eq:pi_0} for $\phi(x) = x^2$ and  $N = 1000$. (a) $\beta = 5 $. (b) $\beta = 0.2$. }
\label{fig:exact_gen2_x2}
\end{figure}
\end{example}
\begin{example}
We consider $\phi(x) = 1-x^2$ to explore the impact of rebelliousness on our model. Since $\phi(x)$ is concave on $[0,1]$ Theorem \ref{homogeneous_dominance} does not apply. In this case, the rate an agent changes his/her opinion  is $ \frac{N^2-n^2}{N^2} + \beta$. Hence, as the number of  agents with the opposite opinion increases, the rate of opinion change decreases. 
%\begin{eqnarray*}
%\label{eq:rate-function-general_(1-x2)}
%	\bar{\lambda}(x) &=&  (1-x)(1-x^2) + \beta (1-x) ,	\nonumber \\
%	\bar{\mu}(x) &=&    x(1-(1-x)^2) +  \beta x.
%\end{eqnarray*} 	
The limiting ODE for this model is
\begin{equation}
\label{eq:binary-ode_(1-x2)}
	\dot{x}(t) =\big( x(t)^2 - x(t) - 1 - \beta \big) \big( 2x(t)-1 \big).
\end{equation}
One can see that \eqref{eq:binary-ode_(1-x2)} has a unique equilibrium at
$\overline{x} = \frac{1}{2}$ which is stable regardless of the spontaneity coefficient $\beta
> 0$. Thus, we can conclude that the model leads to balance of opinions as 
can also be observed in Fig.~\ref{fig:exact_gen_(1-x2)}.
\begin{figure}[h]
\centering
\begin{subfigure}{.5\textwidth}
  \centering
  \includegraphics[width=1\linewidth]{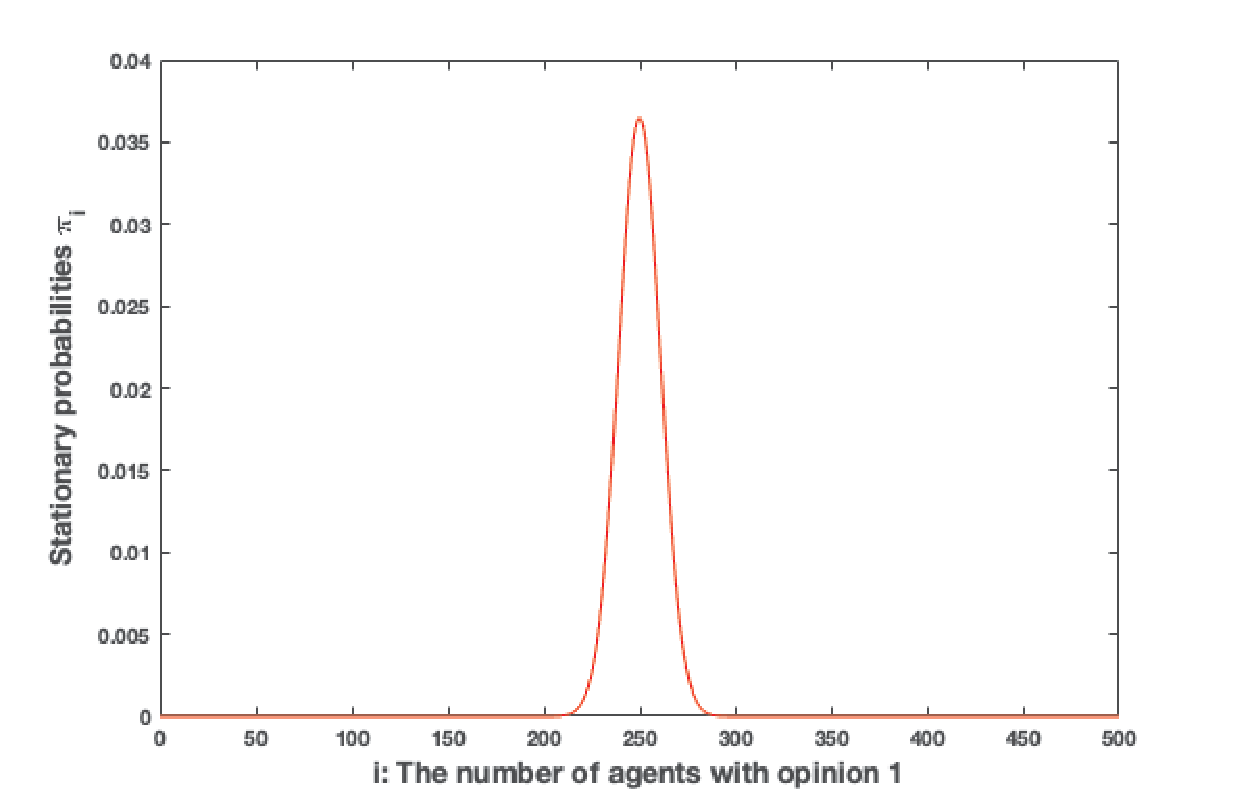}
  \caption{}
  \label{fig:exact500_(1-x2)}
\end{subfigure}%
\begin{subfigure}{.5\textwidth}
  \centering
  \includegraphics[width=1\linewidth]{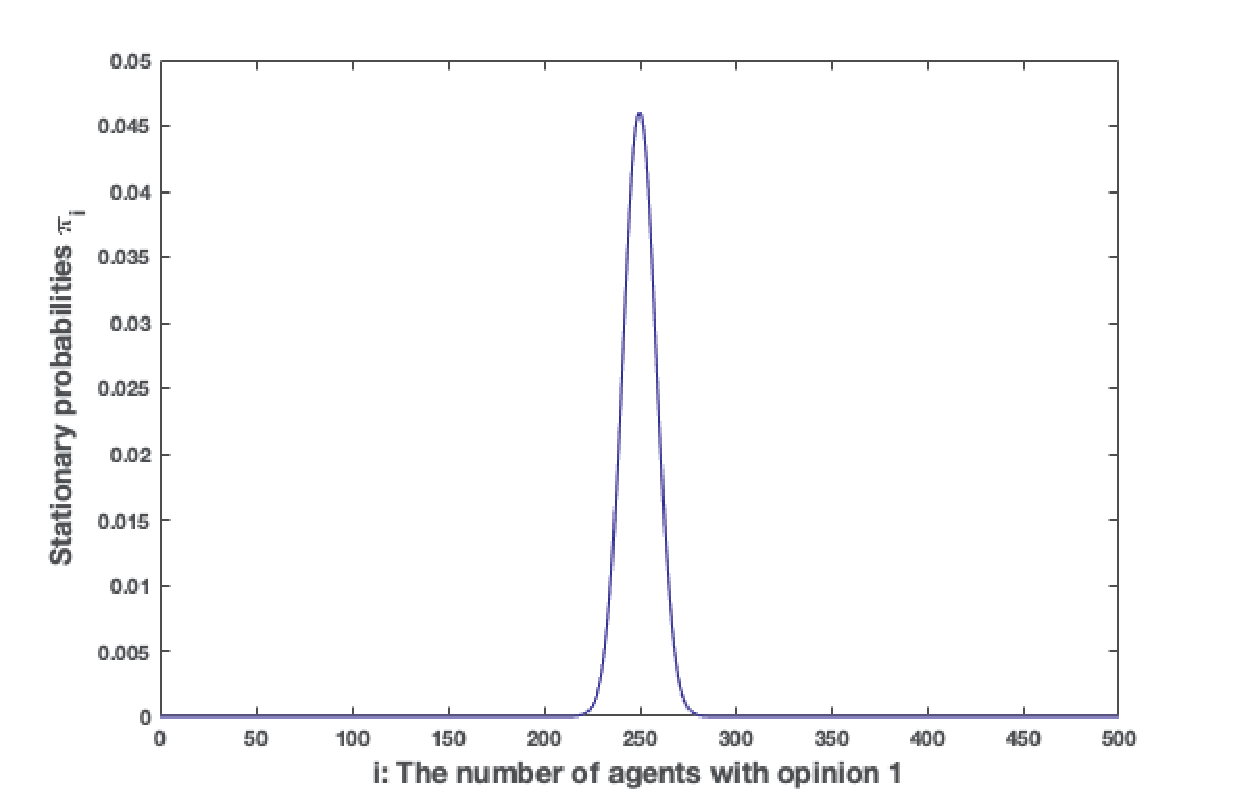}
  \caption{}
  \label{fig:exact500_(1-x2)_second}
\end{subfigure}
 \caption{The stationary probabilities calculated using \eqref{eq:pi_n} for $\phi(x) = 1-x^2$ and and $N = 500$. (a) $\beta = 7$. (b) $\beta = 0.001$.}
\label{fig:exact_gen_(1-x2)}
\end{figure}
\end{example}
\begin{example}
 Let  $\phi(x) = \sqrt{x}$ which is concave. Then 
%\begin{eqnarray*}
%\bar{\lambda}(x) &=& (1-x)\sqrt{x} + \beta(1-x),\\
%\bar{\mu}(x) &=& x \sqrt{1-x} +\beta x.
%\end{eqnarray*}
the corresponding ODE is 
\begin{equation}
\label{eq:binary-ode_(sqrtx)}
  \dot{x}(t) = (1-2x(t))  \left( \frac{ \sqrt{x(t)}\sqrt{1-x(t)} }{ \sqrt{1-x(t)}+\sqrt{x(t)} }  + \beta \right) .
\end{equation}
We first observe that $F$ is not $C^1$ on $[0,1]$ and the fluid limit theorem 
does not apply. Nevertheless, we proceed heuristically to look for the stable equilibria. 
One can observe that $\overline{x} = 1/2$ is the only equilibrium for  \eqref{eq:binary-ode_(sqrtx)}.
%Analyzing the sign of $F$ before and after the equilibrium, we can reach that $1/2$ is stable. 
Thus, one may expect that the model will lead to balance of opinions regardless of the choice of the spontaneity coefficient $\beta$ as can be observed in Fig.~\ref{fig:exact_gen_sqrt}.
\begin{figure}[h]
\centering
\begin{subfigure}{.5\textwidth}
  \centering
  \includegraphics[width=1\linewidth]{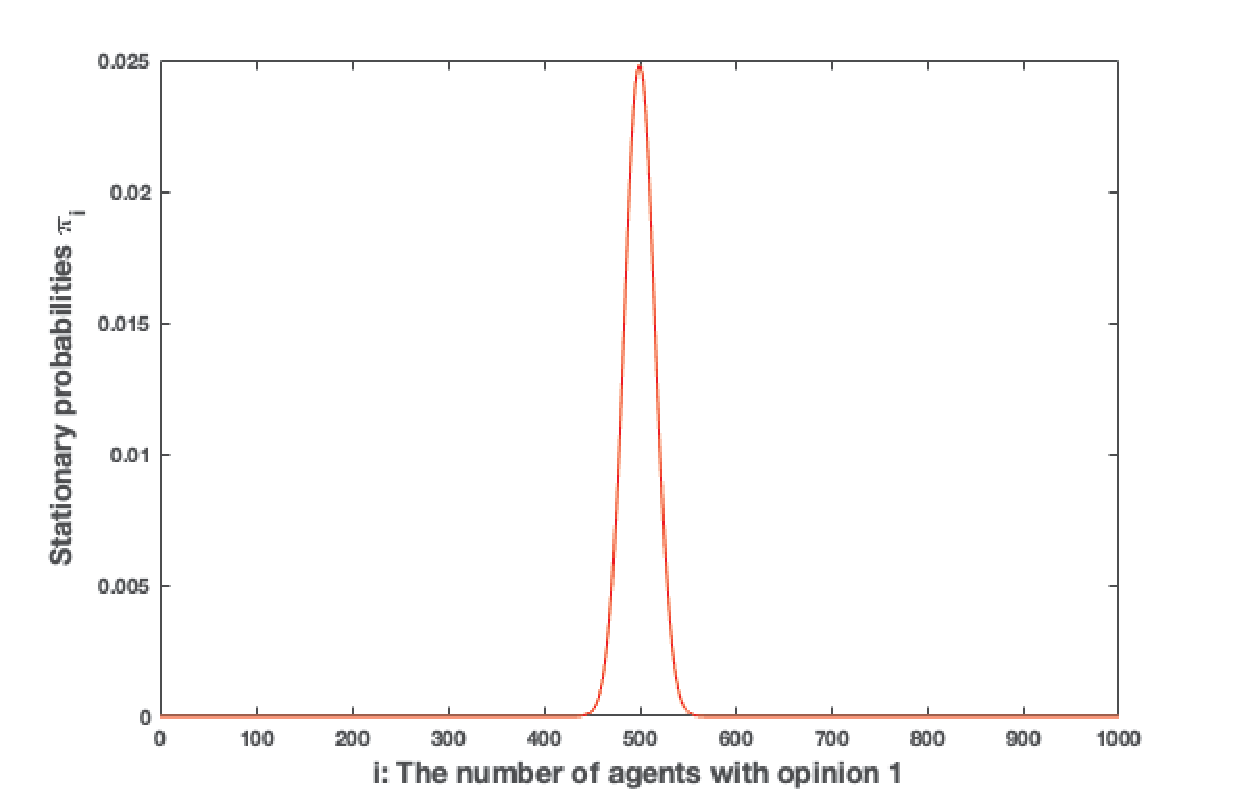}
  \caption{}
\end{subfigure}%
\begin{subfigure}{.5\textwidth}
  \centering
  \includegraphics[width=1\linewidth]{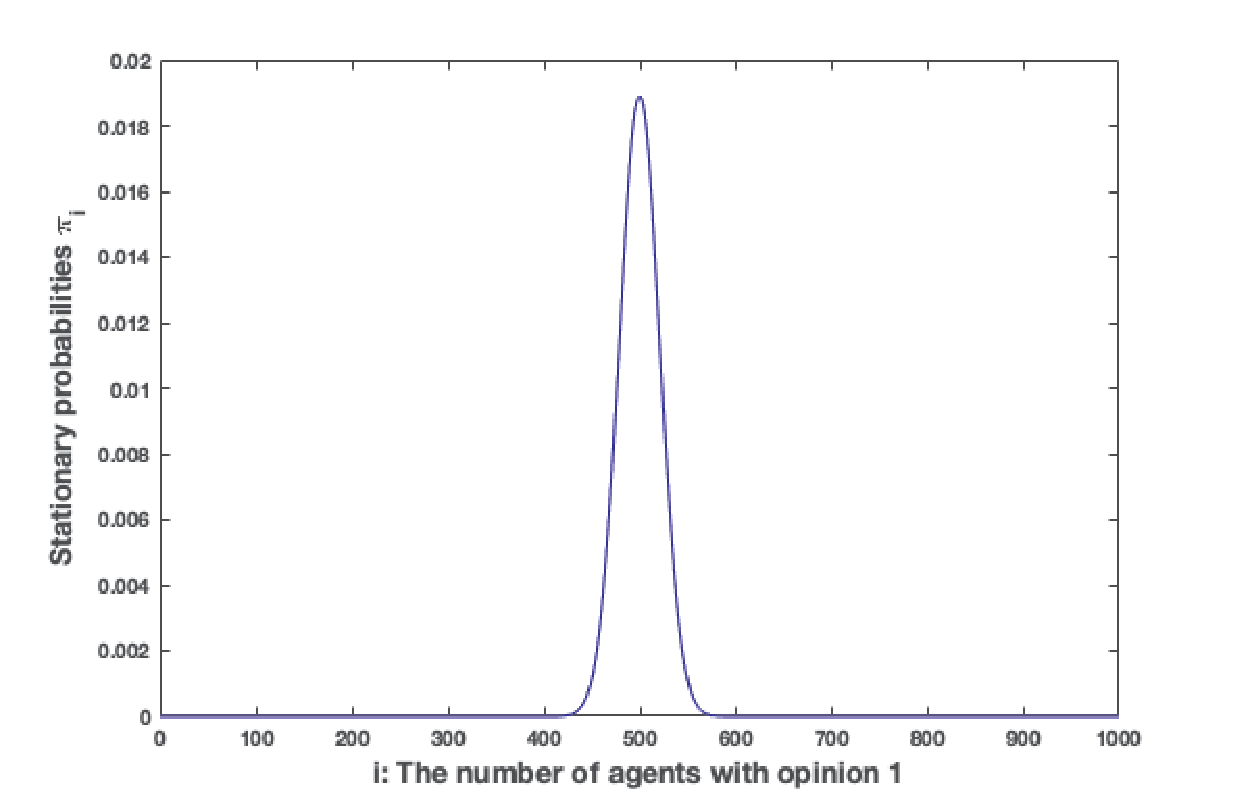}
  \caption{}
\end{subfigure}
\caption{The stationary probabilities calculated using \eqref{eq:pi_n} and  \eqref{eq:pi_0} for $\phi(x) = \sqrt{x}$ and  $N = 1000$. (a) $\beta = 10 $. (b) $\beta = 0.1$. }
\label{fig:exact_gen_sqrt}
\end{figure}

%We should note that our numerical experiments (not shown here) suggests that
%when the personality function, $\phi(x)$,  is strictly convex on $[0,1]$
%(e.g. $x^2, x^3$,..., etc.), the spontaneity coefficient $\beta$ can be chosen
%small enough that the model leads to dominance of one opinion.
\end{example}
\begin{example}
Consider the monotone increasing, convex conformity
function $\phi(x) = \frac{x}{1-x}$, so that the conformity function is the
ratio of the fraction of agents with opposite opinion to those with the same opinion. Then, the rate of change of one's opinion is 
$\frac{n}{N-n} +\beta$.
We note that $\phi(x)$ has a singularity at $x=1$ and that \eqref{eq:bd-rates} does not hold for $i = 0,N$, 
and hence the fluid limit theorem does not hold. Nevertheless, we proceed 
heuristically to 
consider $F$ in \eqref{eq:vectorfield-homo} which is given by
\[
F(x) = (\beta -1)(1-2x),
\]
which has only one equilibrium $\overline{x}=1/2$ and it is (asymptotically) stable if and only
if $\beta >1$. Thus we expect balance for $\beta>1$. When $\beta <1$, we 
expect dominance of one opinion. This heuristic is verified by our numerical computations of the stationary
probabilities \eqref{eq:pi_n} and \eqref{eq:pi_0} as shown in 
 Fig.~\ref{fig:discontinuous}.
\begin{figure}[h]
\centering
\begin{subfigure}{.5\textwidth}
  \centering
  \includegraphics[width=1\linewidth]{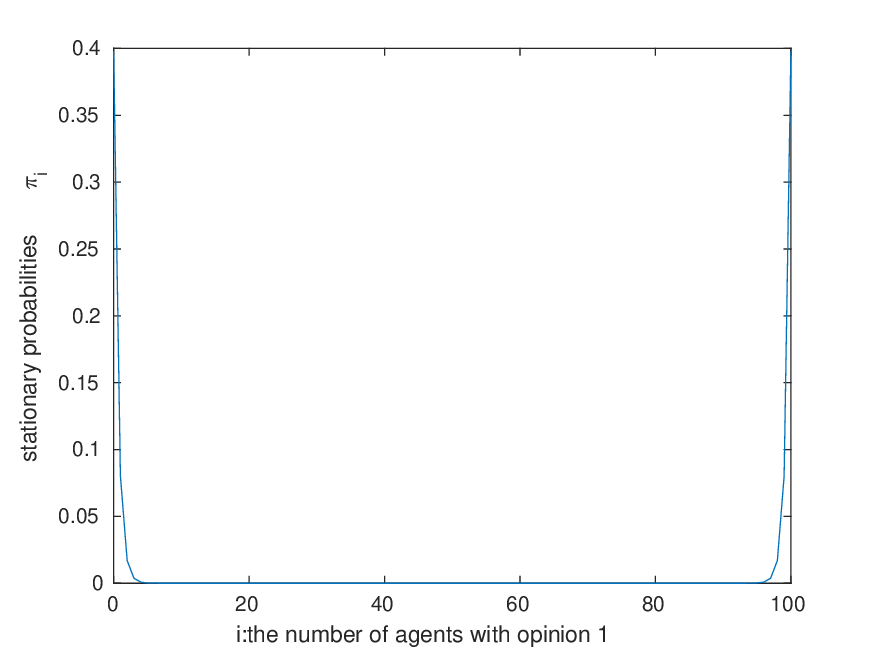}
  \caption{}
\end{subfigure}%
\begin{subfigure}{.5\textwidth}
  \centering
  \includegraphics[width=0.9\linewidth]{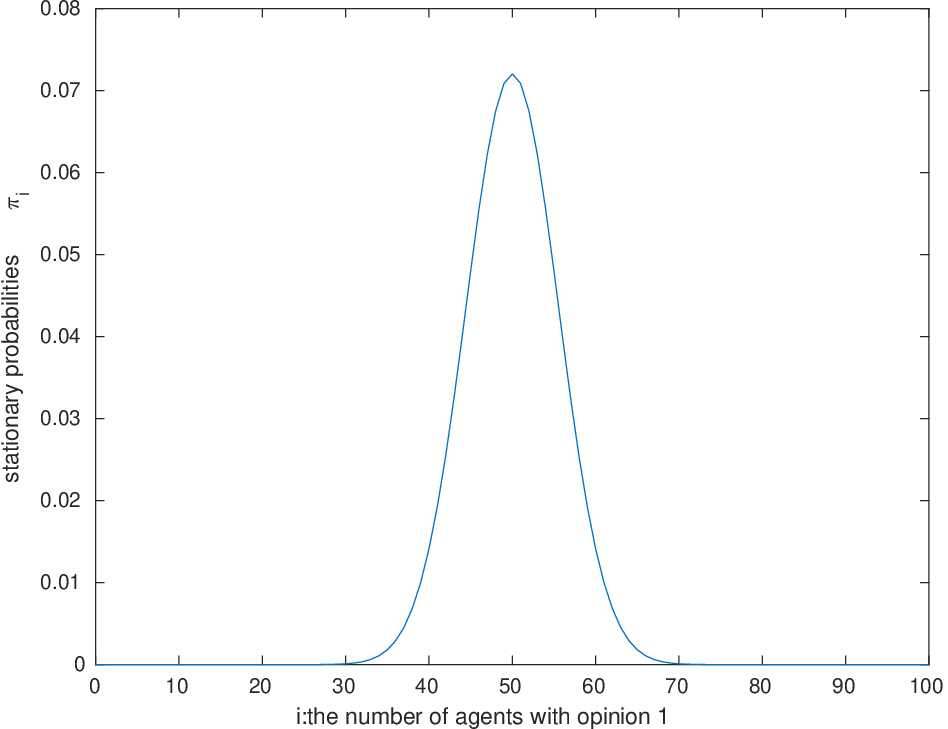}
  \caption{}
\end{subfigure}
\caption{The stationary probabilities calculated using \eqref{eq:pi_n} and  \eqref{eq:pi_0} for 
$\phi(x) = \frac{x}{1-x}$ and $N = 100$ (a) $\beta = 0.2$. (b) $\beta = 10$. }
\label{fig:discontinuous}
\end{figure}
\end{example}

    %! Author = sbbfti
%! Date = 10/06/2020

\section{\label{sec:heterogeneous}Heterogeneous binary opinion dynamics}
Let us consider the case where the group is \emph{heterogeneous}. Namely, suppose we have $m$ personality classes of agents such that all agents 
within a Class $i$ (where $i=1,\dots, m$) have the same personality
$(\phi_i,\beta_i)$, but personalities differ among the classes. 
This results in a Markov process model where the state is 
a vector $x=(x_1,\dots, x_m)$ where $0 \leq x_i \leq N_i$ is the
number of Class $i$ agents who hold opinion $1$ with $N_i$ being the 
total number of Class $i$ agents. We assume that the personalities of agents 
is fixed in time, thus $N_i$ is a constant for each $i$ and $N=N_1+\dots+N_m$
is the total number of all agents. We note that during a time interval
$(t,t+h]$ an agent from Class $i$ will flip with probability
\[
(\phi_i(n/N) + \beta_i)h + o(h) \;\; h \to 0+,
\]
where $n$ is the total number of all agents who have the opposite opinion 
to that of the given agent. We shall be concerned with the case of large $N$ 
with the fractions $k_i=N_i/N$ within classes being held constant. 
This results in a family of Markov process $X^N(t)$ which undergo
a jump $e_i$ or $-e_i$ for $i=1,\dots,m$ (here $e_i \in \real^m$ is the vector
with $i$th component equal to one and all others equal to zero) with corresponding Class $i$ birth and death rates given by
\begin{equation}\label{eq-hetero-lambdamu}
\begin{aligned}
\lambda^N_i(x) &= \phi_i\left(\frac{|x|}{N}\right) (N_i-x_i) + \beta_i (N_i - x_i),\\   
\mu^N_i(x) &= \phi_i\left(1-\frac{|x|}{N}\right) x_i + \beta_i x_i,   
\end{aligned}
\end{equation}
where given the state $x=(x_1,\dots,x_m)$ we denote by $|x|$ the total 
number of agents holding opinion $1$:
\[
|x| = \sum_{i=1}^m x_i.
\]
We note that $0 \leq x_i \leq N_i$. We shall consider the normalized process $X_N(t)=X^N(t)/N$, 
where $X_{N,i}(t)$ is the fraction of Class $i$ agents with opinion $1$ where the fraction is
normalized by $N$ and not $N_i$. We may write
\[
\begin{aligned}
\lambda^N_i(x) &= N \bar{\lambda}_i\left( \frac{x}{N}\right), \\   
\mu^N_i(x) &= N \bar{\mu}_i \left( \frac{x}{N} \right)    
\end{aligned}
\]
where 
\begin{equation}\label{eq-hetero-barlambdamu}
\begin{aligned}
\bar{\lambda}_i(x) &= \phi_i(|x|) (k_i-x_i) + \beta_i (k_i-x_i),\\   
\bar{\mu}_i(x) &= \phi_i\left(1-|x|\right) x_i + \beta_i x_i,    
\end{aligned}
\end{equation}
and as before $|x|=x_1 + \dots x_m$.

In the fluid limit, as $N \to \infty$, one expects $X_N$ to converge to $x$ 
where $x$ satisfies the ODE 
\[
\dot{x}(t)= F(x(t)),
\]
where the $m$ dimensional vector field $F$ is given by
\[
F_i(x) = \bar{\lambda}_i(x)-\bar{\mu}_i(x), \quad i=1,\dots,m,
\]
which simplifies to
\begin{equation}\label{eq_hetero_F}
F_i(x) = \beta_i (k_i - 2 x_i) + \phi(|x|) (k_i-x_i) - \phi(1-|x|) x_i
\end{equation}
for $ i=1,\dots,m.$ When $N$ and $t$ are both large, we expect to see the peaks of the probability 
distribution of $X_N(t)$ to occur near the stable equilibria of this ODE. 
We note that $\bar{x}=(k_1/2,\dots,k_m/2)$ is always an equilibrium.

    \section{The case of two extreme personality classes\label{sec:twoPersonality}}

We consider a heterogeneous group formed with two extreme personality classes $(\phi_i,\beta_i)$ for $i=1,2$ where 
\begin{equation}\label{eq_two_types}
\begin{aligned}
\phi_1(\xi) &= \phi(\xi), & \beta_1&=0,\\
\phi_2(\xi) &= 0,  & \beta_2&=\beta>0,\\ 
\end{aligned}
\end{equation}
where  $\phi: [0,1] \to \real$ is a monotonic function.
We note that Class $1$ corresponds to total conformity and Class $2$ 
corresponds to total spontaneity. Let us write $k_1= 2k$ (thus $k$ is half the
fraction of Class $1$) and thus $k_2=1- 2k$. This results in 
\begin{equation}\label{eq:vectorfields}
\begin{aligned}
F_1(x) &= \phi(x_1+x_2)(2k-x_1) - \phi(1-x_1-x_2) x_1,\\
F_2(x) &= \beta (1-2k - 2 x_2).
\end{aligned}
\end{equation}
At an equilibrium, clearly $x_2=\frac{1}{2}-k$ and $ \bar{x}=(k,1/2-k)$ is always an equilibrium. Additional equilibria are found by solving the equation
\begin{equation}\label{eq:eq-x1bar}
 \phi(x_1 + \frac{1}{2}-k) (2k-x_1) - \phi(\frac{1}{2}-x_1+k) x_1 =0
\end{equation}
for $x_1$. We note that Class 2 (spontaneous class) is always expected to reach a balance since at an equilibrium $x_2 = \frac{1-2k}{2}=\frac{k_2}{2}$. 

\begin{thm}
Consider the group with two extreme personality classes  \eqref{eq_two_types} such that the conformity function $\phi(x) : [0,1] \to [0,\infty)$ is monotone, $C^1$ and 
$\phi'(x)$ is strictly increasing on $[0,1]$. When the fraction of the class of conformists   $2k > \frac{2\phi(1/2)}{  \phi'(1/2)}$,  the equilibrium $\bar{x}= (k,1/2-k)$  is unstable and hence  the model \eqref{eq_two_types} leads to dominance of one opinion for large $N$ and large $t$. 
\end{thm} 
\begin{pf}
The Jacobian at the equilibrium,
\[
	J(k, 1/2-k) = 
  \begin{bmatrix}
    \frac{\partial F_1}{\partial x_1} & \frac{\partial F_1}{\partial x_2}  \\
    \frac{\partial F_2}{\partial x_1} & \frac{\partial F_2}{\partial x_2}.
  \end{bmatrix} =
  \begin{bmatrix}
  2k \phi'(1/2) -2 \phi(1/2) & 2k\phi'(1/2) \\
                  0          & -2\beta
  \end{bmatrix}.
\]
Hence the eigenvalues are
\[
\epsilon_1 = 2k \phi'(1/2) -2 \phi(1/2), \quad   \epsilon_2 = \frac{\partial F_2}{\partial x_2} =  -2\beta<0.
\]
Thus, when $2k >  \frac{2\phi(1/2)}{  \phi'(1/2)}$, the equilibrium $\bar{x}= (k, 1/2-k)$ is unstable and the model leads to dominance. Note that
$\phi'(1/2) \neq 0$.  Otherwise, since 
$\phi'(x)$ is strictly increasing on $[0,1]$, one can conclude that
 $\phi'(x) <0 $ on $[0,1/2)$ and $\phi'(x) >0$ on $(1/2, 1]$ which  contradicts the monotonicity of $\phi(x)$ on $[0,1]$.
%Since Class 2 (class of total spontaneity) is always expected to reach balance for large N and t, we only need to 
 %study the limiting behavior of the Class 1 (class of conformists) for large $N$ and $t$. We focus on the dynamics of the conformists below and show that it has equilibrium points other than $x_1 =k$ which is unstable.
%\begin{equation}
%\label{conformistsODE}
%\dot{x}_1 = H(x_1) =  \phi(x_1 + \frac{1}{2}-k) (2k-x_1) - \phi(\frac{1}{2}-x_1+k) x_1.
%\end{equation}
%Since $ x_2 = 1/2-k$ and $ 0 \leq x_1+x_2 \leq 1$,  $0 \leq x_1 \leq \frac{1}{2}+k$ where $0 < k < 1/2$.  
%Moreover, 
%\[ H(0) = 2k\phi(\frac{1}{2} - k) >0, \quad 
%H(1/2+k) = (k-\frac{1}{2}) \phi(1) <0. \]
%However, $x_1 = k$ is an unstable equilibrium point. Thus, there are two other equilibrium points for \eqref{conformistsODE} that are located on the intervals  $(0,k)$  and   $(k,1/2+k)$.
%Generically, these equilibria will be stable. Hence, 
%conformists and thus the entire group will reach dominance of one opinion for large $N$ and $t$. 
\end{pf}
\begin{example} 
Consider $\phi(x) = x$. The corresponding ODE has a unique equilibrium $(k, 1/2-k)$ and  it is asymptotically stable ($\epsilon_1 = 2k - 1$, $\epsilon_2 = -2\beta$). Thus, regardless of the
choice of $\beta$ and the fraction of the class of conformists, $2k$, this model leads to balance within both classes. Therefore,  the whole group reaches balance.
\end{example}
\begin{example}
Let  $\phi(x) = x^2$. In this case, when  conformists are less than $50\% $ of the population, $2k < \frac{1}{2}$, the corresponding ODE has only one equilibria $(k, 1/2-k)$ and it is asymptotically stable $( \epsilon_1 = 2k-1/2)$. Thus, conformists reach balance as well as the spontaneous class, and hence the entire group reaches balance for large $N$ and large $t$. On the other hand, when conformists are more than $ 50\% $ of the population, $2k >\frac{1}{2}$, the equilibrium  $(k, 1/2-k)$ is unstable and the class of conformists reaches dominance of one opinion. In fact, the other two equilibria are $(k \pm \frac{\sqrt{4k-1}}{2}, 1/2-k )$ and 
 the stability analysis suggests that these equilibrium points are stable. Hence, the model leads to dominance for the whole group. One example is given in Fig.~\ref{fig:TwoTypes} where probabilities are computed using Monte Carlo simulations of $10.000$ trajectories. Here, the total number of agents is $N = 120$ with $N_1=100$ ($2k = 5/6$) being the population of the conformists where $\phi(x) = x^2 $ and spontaneity coefficient $\beta = 0.02$.
\begin{figure}
\begin{center}
\includegraphics[width =8.4cm]{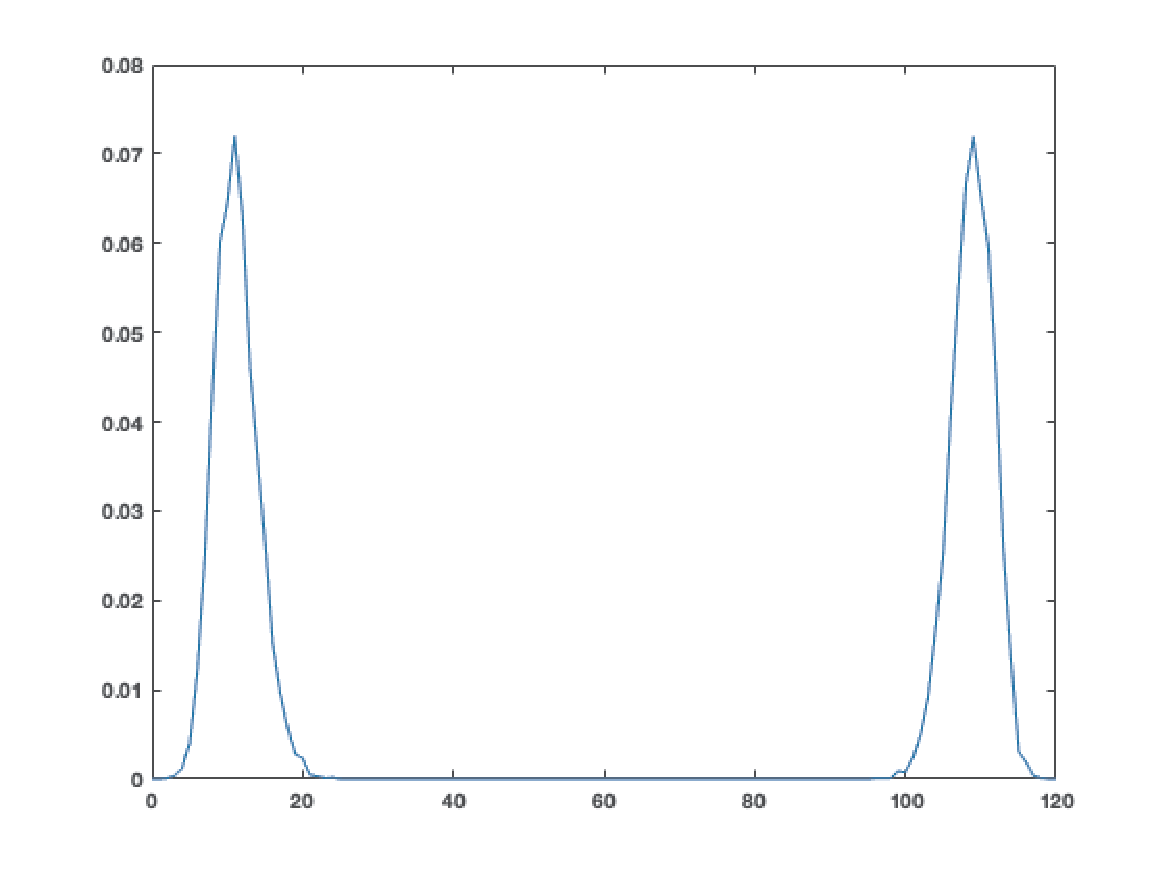} 
\caption{Empirical probability mass function for the number of agents with opinion 1 at time $T=100$ for the case of two extreme classes with $\phi(x)=x^2$ and  $\beta = 0.02$.}
\label{fig:TwoTypes}
\end{center}
\end{figure}
\end{example}

    \section{Conclusions}
\label{sec:conclusions}
We have proposed a simple binary model where agents hold an opinion from the
set $\{0,1\}$ at any time $t \geq 0$. An agent flips his/her opinion based on the opinion distribution of the entire group and his/her personality. We define personality of an agent by 
a monotonic conformity function $\phi$ and a spontaneity coefficient  $\beta
$. When all agents in the group share the same personality, we call the group
homogeneous. 

Initially, focusing on a homogeneous group, we analyzed the
long time probabilities for a large population size for different personality characteristics of the group. The question
of what personality characteristics lead to dominance of one opinion was
studied. We found that the shape of the conformity function, namely strict
convexity or lack thereof, seems to be an important determining factor in 
whether dominance of one opinion occurs for sufficiently small $\beta$.  

We extended  our model to a heterogeneous group, where the group consists of
different personality classes. In particular, when the group is formed by two
extreme classes, complete conformity and complete spontaneity, the dominance
of group opinion is analyzed. In this
example, we found that the fraction of the pure conformists was a key 
determining factor of dominance along 
with the strict convexity of $\phi$. 

    \bibliography{references}

\end{document}